\documentclass[12pt]{iopart}

\usepackage[latin1]{inputenc}
\usepackage{iopams} 
\usepackage{amsthm}
\usepackage{lscape}
\usepackage{graphicx}
\usepackage{makeidx}
\usepackage{mathtools}
\usepackage[caption=false,font=footnotesize]{subfig}
\usepackage{pgfplots}
\usepackage{csquotes}
\MakeOuterQuote{"}

\newtheorem{lemma}{Lemma}[section] 
\newtheorem{theorem}{Theorem}[section]

\theoremstyle{remark}
\newtheorem{remark}{Remark}[section]
\newtheorem{example}{Example}[section]
\theoremstyle{definition}
\newtheorem{definition}{Definition}[section]

\newcommand\norm[1]{\left\lVert#1\right\rVert}

\newcommand\bracket[1]{\langle#1\rangle}

\makeindex
\begin{document}
\title{Quantum filtering for multiple diffusive and Poissonian measurements}
\author{Muhammad F. Emzir, Matthew J. Woolley, Ian R. Petersen}
\address{School of Engineering and Information Technology, University of New South Wales, ADFA, Canberra, ACT 2600, Australia}
%\ead{submissions@iop.org}
\vspace{10pt}
\begin{indented}
\item[]November 2014
\end{indented}

	\begin{abstract}
		We provide a rigorous derivation of a quantum filter for the case of multiple measurements being made on a quantum system. We consider a class of measurement processes which are functions of bosonic field operators, including combinations of diffusive and Poissonian processes. This covers the standard cases from quantum optics, where homodyne detection may be described as a diffusive process and photon counting may be described as a Poissonian process. We obtain a necessary and sufficient condition for any pair of such measurements taken at different output channels to satisfy a commutation relationship. Then, we derive a general, multiple-measurement quantum filter as an extension of a single-measurement quantum filter \cite{bouten2007introduction}. As an application we explicitly obtain the quantum filter corresponding to homodyne detection and photon counting at the output ports of a beam splitter, correcting an earlier result \cite{Kuramochi2013}.
	\end{abstract}
	
	\section{Introduction}
	An optimal filter provides the best estimate of unknown variables through a set of observations of a system. To construct the filter, we need three key ingredients. The first is the probability law corresponding to an observation event. The second is the conditional expectation, which relates an  observation result to the unknown variables. Finally, we need to construct the stochastic differential equations, which describe the estimation result.
\\
	The quantum filtering problem was considered in the early 1980's in a series of articles by Belavkin \cite{Belavkin1992,belavkin1980quantum,belavkin1989nondemolition}. In quantum mechanics, any two random variables (represented by operators) do not always commute. This fact requires an extension of Kolmogorov's classical probability theory  to the non-commutative probability theory used in quantum filtering. In the theoretical physics community, the quantum filtering problem is known under the names of the stochastic master equation and quantum trajectory theory \cite{carmichael1993open,wiseman2010quantum}.
	\\
	Quantum filters are typically derived for the case of single measurements. The quantum filtering problem with multiple output fields has been developed using quantum trajectory theory in Refs. \cite{wiseman2010quantum,wiseman2001complete}  with application to multiple-input multiple-output (MIMO) quantum feedback \cite{chia2011quantum}. In Ref. \cite{nurdin2014quantum}, the multiple-output measurements of the generalized "dyne"  type were also considered for the case of a zero-mean jointly Gaussian state.	It is desirable to extend these results to cover a wider class of measurements, e.g., to include both homodyne detection and photon counting in a quantum optics setting. However, in the case of multiple measurements, it is not clear whether every possible combination of measurements  will satisfy the commutation relations required for a joint probability density of the multiple measurements to exist.
	\\
	A jump-diffusive quantum trajectory has also been derived in Ref. \cite{pellegrini2010markov} using a classical Markov chain approximation of the environmental field, where the system is assumed to interact with the environment over a small time interval.  The infinitesimal generator of the Markov process is obtained as the limiting case when the interaction time goes to zero.
	Recently Ref. \cite{amini2014} proved that for a general class of stochastic master equations (SME) driven simultaneously by Wiener and Poissonian process, the quantum filter possesses sub-martingale properties for the fidelity between the estimated state and the actual state. Our quantum filter's SME descriptions will also fall within a class of the SMEs considered in Ref. \cite{amini2014}. However, we derive the filter using quantum stochastic calculus for a class of measurement processes which are functions of bosonic field operators, including combinations of diffusive and Poissonian processes. 
	\\
	Within the experimental quantum optics community, simultaneous measurement of quantum systems are frequently performed \cite{Broome15022013,Spring15022013,lang2013correlations}. Previously,  non-classical states of light have been reconstructed via post-processing of homodyne detection measurement data \cite{NeergaardNielsen2008,KatanyaB.Kuntz2014}. The thousands of homodyne detection records triggered by photon counts were sampled to construct a Wigner function using a heuristic time window approach. One could replace this procedure with the more systematic quantum filtering approach for multiple measurements that we have derived. In addition to this, the quantum filtering approach using both homodyne and photon counting detections could possibly be used as a solution to the  number-resolved photon counting problem, \cite{chen2011microwave}.
	\\
	The purpose of this article is to derive using quantum stochastic calculus, the quantum filter corresponding to multiple measurements made on a quantum system. To achieve this, we first investigate the commutativity of multiple measurement processes \cite{belavkin1989nondemolition}. We use the definitions of concatenation and series product \cite{gough2009series} to describe quantum systems composed of multiple interacting open quantum systems, each of them described by $(S,L,H)$ parameters \cite{parthasarathy2012}. We then formulate a general quantum filter for a quantum system with a finite number of commutative measurements.
	\\
	We will show that the quantum filter for  a quantum system with multiple measurement outputs can be described by a stochastic master equation for the conditional density operator as follows,
			\begin{align}
			d\rho_t = & \underbrace{\left[-i\left[H_t,\rho_t\right] + \mathbf{L}^{\top}\rho_t\mathbf{L}^{\ast} - \frac{1}{2}\mathbf{L}^{\dagger}\mathbf{L}\rho_t - \frac{1}{2}\rho_t\mathbf{L}^{\dagger}\mathbf{L}\right]dt }_{\text{a-priori}}+ \underbrace{\zeta^{\top}\Sigma^{-\top}d\mathbf{W}}_{\text{innovation term}}.
			\end{align}
	This equation includes an a-priori part which is the original unconditional quantum master equation, and a stochastic part which is contained in the \emph{innovation term}. The innovation term relates the measurement records to the evolution of the conditional density operator. In this equation, $d\mathbf{W}$ is a vector of the "error" between the actual measurement and the expected value. $\zeta^{\top}\Sigma^{-\top}$ is the weighting function which associates the contribution of each measurement to the total increment of the conditional density operator. Generally, the evolution of the conditional density operator above will contain both diffusion and jump processes, which allows for simultaneous photon counting and homodyne detection.
	\\	
	We apply our filter to the simple case of photon counting and homodyne measurement at the outputs of a single beam splitter \cite{carmichael2000giant}. \Fref{fig:DoubleMeasurement} shows a typical arrangement of photon counting and homodyne detection at the output  ports of a beam splitter in a quantum optics experiment. In Ref. \cite{Kuramochi2013}, simultaneous continuous measurement of photon counting and homodyne has been considered. However, their derivation does not appropriately account for the presence of the beam splitter. We correct the result of Ref. \cite{Kuramochi2013}, and then give our result for the photon counting and homodyne detection filter in the form of an unnormalized stochastic Schr\"{o}dinger equation (SSE). 
	\\
	We refer the readers to Ref. \cite{bouten2007introduction} for background material, such as an introduction to quantum stochastic calculus, quantum probability, and quantum non-demolition measurements. Furthermore, without loss of generality, the field is assumed to initially be in the vacuum state and the reduced Planck constant $\hbar$ is set to one. We will assume that the quantum stochastic differential equation (QSDE) parameters $S,L,H$ are bounded to ensure that the corresponding solution is unique and unitary, as well as reducing the technical difficulties that  arise.
	
		\begin{figure}
		\centering
		\includegraphics[width=\textwidth]{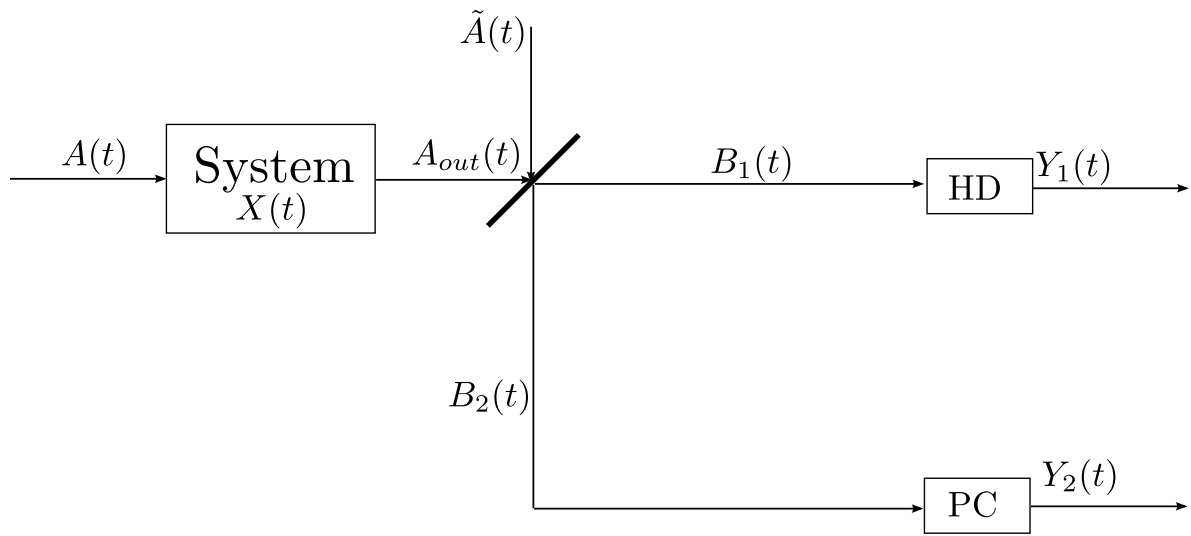}
		\caption{Simultaneous photon counting and homodyne detection at the outputs of a beam splitter in a quantum optics experiment.}
		\label{fig:DoubleMeasurement}
		\end{figure}
	
	\section{Preliminary}
	\subsection{Notation}
	Von Neumann algebras and $\sigma$-algebras are written in calligraphic symbols (e.g. $\mathcal{B}$  for the Borel $\sigma$-algebras on $\mathbb{R}$). As usual, classical probability spaces are written as $(\Omega, \mathcal{F},\mu)$. $\mathbb{E}_{\rho}$ will be used for the expectation of an observable with respect to the density matrix $\rho$. Plain capital letters (e.g. $P$) will be used to denote elements of von Neumann algebras. Bold letters (e.g. $\mathbf{X}$) will be used to denote a matrix whose elements are Hilbert space operators. Serif symbols (e.g. $\mathsf{H}$) are used for Hilbert spaces. Hilbert space adjoints, are indicated by $^{\ast}$, while the complex conjugate transpose will be denoted by $\dagger$, i.e. $\left(\mathbf{X}^{\ast}\right)^{\top} = \mathbf{X}^{\dagger}$. For single-element operators we will use $*$ and $\dagger$ interchangeably. The Hilbert space inner product of $X$ and $Y$ is given by $\bracket{X,Y}$. The commutator of $X$ and $Y$ is given by $[X, Y ] = XY - YX$.

	\subsection{Multiple Output and Input Channels Open Quantum System}
	The dynamics of an open quantum system with multiple bosonic field  input and output channels  can  be described via annihilation, creation and conservation processes. First, let ${z_k, k \geq 1}$ in $\mathsf{V}$ be a complete orthonormal basis. A single particle Hilbert space $\mathsf{h}$ is defined as $\mathsf{V}\otimes L^2(\mathbb{R}_+)$. The \emph{symmetric} Fock space over $\mathsf{h}$ is given by
	\begin{align}
	\Gamma(\mathsf{h}) = \oplus_{n=0}^{\infty}{\mathsf{h}_{\textcircled{s}}}^{n},
	\end{align}
	where ${\mathsf{h}_{\textcircled{s}}}^{n}$ is the $n-$fold symmetric tensor product of $\mathsf{h}$. 
	The exponential vector $e(u) \in \Gamma(\mathsf{h})$ is defined as
	\begin{align}
	e(u) = \oplus_{n=0}^{\infty} \frac{1}{\sqrt{n}} u^{\otimes^n}.
	\end{align}
	The vacuum vector $\Phi$ corresponds to the exponential vector with $u=0$, whilst other coherent vectors $\psi(u)$ are the normalized exponential vectors $e(u), u\neq 0$. The Weyl operator $W(u,U), u \in \mathsf{h}, U \in \mathcal{B}(\mathsf{h})$ is a unique unitary transformation operating on $e(u)$ defined by \cite[\textsection 20]{parthasarathy2012}
	\begin{align}
	W(u,U)e(v) = \left\{\exp\left(-\frac{1}{2}\norm{u}^2 - \bracket{u,Uv}\right)\right\}e\left(Uv+u\right).\label{eq:WeylOperator}
	\end{align}
	\\
	For any $f \in \mathsf{h}$, let us define $f_k (t) \equiv \langle z_k | f(t) \rangle$. The annihilation $A_k (t) $, creation $A^{\dagger}_k (t)$ and conservation $\Lambda^{\dagger}_{kl} (t)$ processes related to the orthonormal basis $z_k$ are given by \cite{barchielli2006continual},
\begin{subequations}
\begin{align}
	A_k (t) & \equiv a (z_k \otimes I_{(0,t]}),\\
	A^{\dagger}_k (t) & \equiv a^{\dagger} (z_k \otimes I_{(0,t]}),\\
	\Lambda^{\dagger}_{kl} (t) & \equiv \lambda (|z_k\rangle\langle z_l| \otimes I_{(0,t]}),
\end{align}\label{eq:QuantumFundamentalProcesses}
\end{subequations}
here the operators $a, a^{\dagger}$ and $\lambda$ is the Stone generator of the corresponding Weyl operators, defined for $H = |z_k\rangle\langle z_l| \otimes I_{(0,t]}$ and $u = z_k \otimes I_{(0,t]}$, 
\begin{align*}
W\left(0,\exp(i t H)\right) = \exp(-it\lambda(H)) ,\; W\left(tu,I\right) = \exp(-itp(u)), \\ 
q = p(iu) ,\; a(u)=\frac{1}{2}\left(q(u)+ip(u)\right) ,\;  a^{\dagger}(u)=\frac{1}{2}\left(q(u)-ip(u)\right).
\end{align*}
		
In the coherent vector domain, the following commutation relations hold,
\begin{subequations}
	\begin{align}
	\left[A_k(t),A_l(s)\right] = &\left[A^{\dagger}_k(t),A^{\dagger}_l(s)\right] = 0 \; , \left[A_k(t),A^{\dagger}_l(s)\right] =  \delta_{kl} \min(t,s).
	\end{align}
\end{subequations}
\\
In addition, other commutation relations can be obtained using the general It\^{o} multiplications as below \cite{parthasarathy2012,barchielli2006continual},
\begin{subequations}
\begin{align}
dt & \equiv d\Lambda_{00} \; , dA_k \equiv d\Lambda_{0k} \; , dA^{\dagger}_k  \equiv d\Lambda_{k0}, \\
d\Lambda_{k} & \equiv d\Lambda_{kk} \; , d\Lambda_{kr}(t)d\Lambda_{sl}(t) = \hat{\delta}_{rs}d\Lambda_{kl},
\end{align}\label{eq:GeneralItoMultiplication}
\end{subequations}
\\
where $\hat{\delta}_{rs}=0, \forall r=0 \cup s=0$ and $\hat{\delta}_{rs} = \delta_{rs}$ otherwise.
The evolution of a system observable in the Heisenberg picture is given by 
\begin{align}
X_t \equiv & U^{\dagger}_t\left(X \otimes I \right)U_t \label{eq:UnitaryMap}
\end{align}
\\
Let $\mathcal{G}$ be an open quantum system with parameters $(\mathbf{S},\mathbf{L},H)$, and $\mathbf{S}\mathbf{S}^{\dagger} = \mathbf{S}^{\dagger}\mathbf{S} = \mathbf{I}$. For any system observable $X$, the following QSDE in the Heisenberg picture is obtained,
\begin{align}
dX_t = & \left(-i\left[X_t,H_t\right] + \mathcal{L}_L (X_t)\right)dt + d\mathbf{A}^{\dagger}_t\mathbf{S}^{\dagger}_t\left[X_t,\mathbf{L}_t\right] + \left[\mathbf{L}^{\dagger}_t,X_t\right]\mathbf{S}_t d\mathbf{A}_t\nonumber\\
& + \text{tr}\left[\left(\mathbf{S}^{\dagger}_t X_t \mathbf{S}_t - X_t\right)d\mathbf{\Lambda}_t^{\top}\right], \label{eq:QSDE_X}
\end{align}
where $\mathcal{L}_L (X_t) = \frac{1}{2} \mathbf{L}_t^{\dagger} \left[X_t,\mathbf{L}_t\right] + \frac{1}{2} \left[\mathbf{L}_t^{\dagger},X_t\right]\mathbf{L}_t$, is the Lindbladian super operator, and all operators evolve according to Eq. \eqref{eq:UnitaryMap}, i.e. $\mathbf{L}_t \equiv U^{\dagger}_t\left(\mathbf{L} \otimes I \right)U_t$. In the Schr\"{o}dinger picture, the corresponding unitary operator evolution is
\begin{align}
					dU_t = &  \left[\text{tr}\left[\left(\mathbf{S} - \mathbf{I}\right)d\mathbf{\Lambda}^\top_t\right] + d\mathbf{A}^{\dagger}_t \mathbf{L}  - \mathbf{L}^{\dagger}\mathbf{S} d\mathbf{A}_t - \left(\dfrac{1}{2}\mathbf{L}^{\dagger}\mathbf{L}+iH\right)dt\right]U_t \;, U_0 = I\label{eq:UnitaryEvolution}
\end{align}
The evolution of the output fields is given by \cite{gough2009series}
\begin{subequations}
	\begin{align}
	d\tilde{\mathbf{A}}_t =& \mathbf{S}_td\mathbf{A}_t + \mathbf{L}_t dt,\\
	d\tilde{\mathbf{\Lambda}}_t =& \mathbf{S}_t^{*} d\mathbf{\Lambda} \mathbf{S}_t^{\top} + \mathbf{S}_t^{*} d\mathbf{A}_t^{\ast}\mathbf{L}_t^{\top} + \mathbf{L}_t^{\ast}d\mathbf{A}_t^{\top}\mathbf{S}_t^{\top} + \mathbf{L}^{\ast}\mathbf{L}^{\top}dt.
	\end{align}\label{eq:OutputFieldEvolution}
\end{subequations}
			
Subsequently, we will show our first result, which will ensure that for a class of output measurements, the commutation relation is satisfied, and hence the corresponding joint probability density function exists.
					
	\section{Main Results}
		
	\subsection{Commutativity of Open Quantum  System Output Channels}
	\begin{definition}\label{def:Commutation}[Commutator of two vectors with non-commutative elements] Let $\mathbf{a},\mathbf{b} \in \mathcal{B}(\mathsf{H})^{n \times 1}$ be vectors whose elements are non-commutative. The commutator of a pair $\mathbf{a},\mathbf{b}$ is given by
			\begin{align}
			[\mathbf{a},\mathbf{b}] = & \mathbf{a}\mathbf{b}^{\top} - \left(\mathbf{b}\mathbf{a}^{\top}\right)^{\top}.
			\end{align}
			\end{definition}
			Notice that in Definition \ref{def:Commutation}, in general, $\left(\mathbf{b}\mathbf{a}^{\top}\right)^{\top} \neq \mathbf{a}\mathbf{b}^{\top}$ due to the non-commutativity of the elements in the pair $\mathbf{a},\mathbf{b}$. We further define the self-commutator of a vector $\mathbf{a}$ whose elements are non-commutative  as $[\mathbf{a},\mathbf{a}]$. It is important to see that the self-commutator is not always equals to zero, as the following simple example shows.
			\begin{example}\label{exm:a-adagger}
			Let $\mathbf{a} = \left[a \; a^{\dagger}\right]^{\top}$,where $a$ is an annihilation operator defined for a Fock space $\Gamma(\mathsf{h})$. Then by Definition \ref{def:Commutation}, the self-commutator of $\mathbf{a}$ is given by
			\begin{align*}
			[\mathbf{a},\mathbf{a}] = & \mathbf{a}\mathbf{a}^{\top} - \left(\mathbf{a}\mathbf{a}^{\top}\right)^{\top}\\
			= & \begin{bmatrix}
			aa & aa^{\dagger},\\
			a^{\dagger}a & a^{\dagger}a^{\dagger}
			\end{bmatrix} - 
			\begin{bmatrix}
					aa & a^{\dagger}a\\
					aa^{\dagger} & a^{\dagger}a^{\dagger}
					\end{bmatrix}
					\\
			= & \begin{bmatrix}
			0 & 1\\
			-1 & 0
			\end{bmatrix}.		
			\end{align*}
			\end{example}
			As the definition and example clearly show, for self-commutator, commutativity is implied by the symmetry properties of $\mathbf{a}\mathbf{a}^{\top}$. Now consider a general measurement equation, which is a function of the field output annihilation, creation, and conservation processes,
			\begin{align}
			d\mathbf{Y}_t =& \mathbf{F}_t^{\ast}d\tilde{\mathbf{A}}_t^{\ast} +\mathbf{F}_t d\tilde{\mathbf{A}}_t+ \mathbf{G}_t d\tilde{\mathbf{\lambda}}_t,\label{eq:GeneralMeasurement}\\
			d\tilde{\mathbf{\lambda}}_t =& \text{diag}\left(d\tilde{\mathbf{\Lambda}}_t\right).\nonumber
			\end{align}
			\\
			Substituting Eq. \eref{eq:OutputFieldEvolution} into \eref{eq:GeneralMeasurement}, we can write the general measurement equation as
			\begin{subequations}
				\begin{align}
				d\mathbf{Y}_t =& \mathbf{F}_t^{\ast} d\mathbf{a}_1+\mathbf{F}_t d\mathbf{a}_2 + d\mathbf{a}_3 + \mathbf{G}_t \left(d\mathbf{b}_1+ d\mathbf{b}_2+d\mathbf{b}_3+d\mathbf{b}_4\right),
				\end{align}
				where
				\begin{align}
				d\mathbf{a}_{1,i} =& \sum_{k = 1}^{n} \mathbf{S}^{*}_{ik}d\mathbf{\Lambda}_{k0} , &
				d\mathbf{a}_{2,i} =& \sum_{k = 1}^{n} \mathbf{S}_{ik}d\mathbf{\Lambda}_{ok} , &
				d\mathbf{a}_{3,i} =& \sum_{k = 1}^{n} \left[\left(\mathbf{FL}\right)^{*}_{i1}+\left(\mathbf{FL}\right)_{i1}\right]dt,\\
				d\mathbf{b}_{1,i} =& \sum_{k,k' = 1}^{n} \mathbf{S}^{*}_{ik}d\mathbf{\Lambda}_{kk'}\mathbf{S}^{\top}_{k'i} , &
				d\mathbf{b}_{2,i} =& \sum_{k = 1}^{n} \mathbf{S}^{*}_{ik}d\mathbf{\Lambda}_{k0}\mathbf{L}^{\top}_{1i} , &
				d\mathbf{b}_{3,i} =& \sum_{k = 1}^{n} \mathbf{L}^{*}_{i1}d\mathbf{\Lambda}^{\top}_{0k}\mathbf{S}^{\top}_{ki},\\
				d\mathbf{b}_{4,i} =& \sum_{k = 1}^{n} \mathbf{L}^{*}_{i1}\mathbf{L}^{\top}_{1i} dt.								
				\end{align}\label{eq:a_and_b}		
			\end{subequations}
			Based on Eq. \eref{eq:GeneralItoMultiplication}, most of the multiplication products between the $d\mathbf{Y}$ elements are zero, while the remaining terms are listed in \Tref{tab:ItoTableResult}. 
			\begin{table}[!h]
				\centering
				\begin{tabular}{|c|c|c|c|}
					\hline  $\times$ &$d\mathbf{b}_1^{\top}$&$d\mathbf{b}_2^{\top}$&$d\mathbf{a}_1^{\top}$\\ 
					\hline  $d\mathbf{b}_1$& $\sum_{k,l,l' = 1}^{n} \mathbf{S}^{*}_{ik}\mathbf{S}^{\top}_{li}\mathbf{S}^{*}_{jl}\mathbf{S}^{\top}_{l'j}d\mathbf{\Lambda}_{kl'}$ & $\sum_{k,l = 1}^{n} \mathbf{S}^{*}_{ik}\mathbf{S}^{\top}_{li}\mathbf{S}^{*}_{jl}\mathbf{L}^{\top}_{1j}d\mathbf{\Lambda}_{k0}$  & $\sum_{k,l=1}^{n} \mathbf{S}_{ik}\mathbf{S}^{\top}_{li} \mathbf{S}^{*}_{jl}d\mathbf{\Lambda}_{k0} $   \\ 
					\hline  $d\mathbf{b}_3$& $\sum_{l,l' = 1}^{n} \mathbf{L}^{*}_{i1}\mathbf{S}^{\top}_{li}\mathbf{S}^{*}_{jl}\mathbf{S}^{\top}_{l'j}d\mathbf{\Lambda}_{0l'}$  & $\sum_{k = 1}^{n} \mathbf{L}^{*}_{i1}\mathbf{S}^{\top}_{ki}\mathbf{S}^{*}_{jk}\mathbf{L}^{\top}_{1j}dt$  & $\sum_{k}^{n} \mathbf{L}^{*}_{i1}\mathbf{S}^{\top}_{ki}\mathbf{S}^{\ast}_{jk} dt$ \\
					\hline  $d\mathbf{a}_2$& $\sum_{k,l'=1}^{n} \mathbf{S}_{ik}\mathbf{S}^{*}_{jk}\mathbf{S}^{\top}_{l'j} d\mathbf{\Lambda}_{0l'}$  & $\sum_{k=1}^{n} \mathbf{S}_{ik} \mathbf{S}^{*}_{jk}\mathbf{L}^{\top}_{1j}dt$  & $\mathbf{I}_{ij}dt$  \\
					\hline 
				\end{tabular}
				\caption{It\^{o} multiplication table for the $d\mathbf{Y}$ components. } \label{tab:ItoTableResult} 
			\end{table}
			\begin{remark}\label{rem:selfCommutativity}
				A set of measurements $\mathbf{Y}_t$ made at the output of a quantum system is \emph{self-commutative} if and only if $d\mathbf{Y}_t d\mathbf{Y}_t^{\top}$ is symmetric.
			\end{remark}
			This fact follows directly from Definition \ref{def:Commutation}. Now we state the following Lemma to prove our main result on the commutation relation for a finite number of outputs of an open quantum system.
			\begin{lemma}\label{lem:ItoTableResult}
				The off diagonal elements in the It\^{o} \Tref{tab:ItoTableResult} for multiplication between the $d\mathbf{Y}$ elements are zero.
			\end{lemma}
			\begin{proof}
				For every entry in \Tref{tab:ItoTableResult}, we have
				\begin{align*}
				\sum_{l = 1}^{n} \mathbf{S}^{\top}_{ki} \mathbf{S}^{*}_{jk} =  \sum_{l = 1}^{n} \mathbf{S}^{\top}_{li} \mathbf{S}^{*}_{jl} = & \left(\mathbf{S}_t\mathbf{S}^{\dagger}\right)_{i,j} = \left(I\right)_{i,j} = 0, \; \forall i \neq j ,
				\end{align*}
				which shows that the non-diagonal elements of the multiplication results are zero.
			\end{proof}
			
			\begin{theorem}\label{thm:CommutativityOfQuantumNetwork}
			A set of $n$ general measurements \eref{eq:GeneralMeasurement} $\mathbf{Y}$ is self-commutative for any multiple-output quantum system with $n$ channels if and only if,
			\begin{subequations}
				\begin{align}
				\begin{bmatrix}
				\mathbf{F} & \mathbf{F}^{\ast}
				\end{bmatrix}
				\begin{bmatrix}
				\mathbf{0} & \mathbf{I}\\
				-\mathbf{I} & \mathbf{0}
				\end{bmatrix}
				\begin{bmatrix}
				\mathbf{F}^{\top} \\ \mathbf{F}^{\dagger}
				\end{bmatrix} = & \mathbf{0}, \label{eq:ConditionOnF}\\
				\begin{bmatrix}
				\mathbf{G} & \mathbf{F}^{\ast}
				\end{bmatrix}
				\begin{bmatrix}
				\mathbf{0} & \mathbf{I}\\
				-\mathbf{I} & \mathbf{0}
				\end{bmatrix}
				\begin{bmatrix}
				\mathbf{G}^{\top} \\ \mathbf{F}^{\dagger}
				\end{bmatrix} = & \mathbf{0},\label{eq:ConditionOnGF_ast}\\
				\begin{bmatrix}
				\mathbf{G} & \mathbf{F}
				\end{bmatrix}
				\begin{bmatrix}
				\mathbf{0} & \mathbf{I}\\
				-\mathbf{I} & \mathbf{0}\label{eq:ConditionOnGF}
				\end{bmatrix}
				\begin{bmatrix}
				\mathbf{G}^{\top} \\ \mathbf{F}^{\top}
				\end{bmatrix} = & \mathbf{0}.
				\end{align}			
			\end{subequations}
			Furthermore, Eq. \eqref{eq:ConditionOnF} is equivalent to  $\Re(\mathbf{F})\Im(\mathbf{F})^{\top}$ being symmetric, Eq. \eqref{eq:ConditionOnGF_ast} is equivalent to both $\Re(\mathbf{G})\Re(\mathbf{F})^{\top},\Im(\mathbf{G})\Im(\mathbf{F})^{\top}$ being symmetric whilst  Eq. \eqref{eq:ConditionOnGF}
			is equivalent to both $\Re(\mathbf{G})\Im(\mathbf{F})^{\top},\Im(\mathbf{G})\Re(\mathbf{F})^{\top}$ being symmetric.
			\end{theorem}
			\begin{proof}
			Let $\mathbf{Y}$ be a generalized measurement whose evolution is given by \eref{eq:GeneralMeasurement}. Then to prove the theorem, Remark \ref{rem:selfCommutativity} shows that it is sufficient to show that $d\mathbf{Y}d\mathbf{Y}^{\top}$ is symmetric in order to show that all measurement outputs commute with each other. Simplifying \Tref{tab:ItoTableResult}  to \Tref{tab:ItoTableResultSimplified} and evaluating the $d\mathbf{Y}d\mathbf{Y}^{\top}$, from \Tref{tab:ItoTableResultSimplified}, one obtains,
			\begin{align}
						\left(d\mathbf{Y}d\mathbf{Y}^{\top}\right)_{ij} = & \left(\mathbf{G}\left[d\mathbf{b}_1d\mathbf{b}_1^{\top}+d\mathbf{b}_1d\mathbf{b}_2^{\top}+d\mathbf{b}_3d\mathbf{b}_1^{\top}+d\mathbf{b}_3d\mathbf{b}_2^{\top}\right]\mathbf{G}^{\top}\right)_{ij} \nonumber\\
						&+\left(\mathbf{G}\left[d\mathbf{b}_1d\mathbf{a}_1^{\top}+d\mathbf{b}_3d\mathbf{a}_1^{\top}\right]\mathbf{F}^{\dagger}\right)_{ij}\nonumber\\
						&+\left(\mathbf{F}\left[d\mathbf{a}_2d\mathbf{b}_1^{\top}+d\mathbf{a}_2d\mathbf{b}_2^{\top}\right]\mathbf{G}^{\top}\right)_{ij}\nonumber\\
						&+\left(\mathbf{F}d\mathbf{a}_2d\mathbf{a}_1^{\top}\mathbf{F}^{\dagger}\right)_{ij} \nonumber\\
						&= \left(\mathbf{G}\mathbf{O}_1\mathbf{G}^{\top} + \mathbf{G}\mathbf{O}_2\mathbf{F}^{\dagger} + \mathbf{F}\mathbf{O}_3\mathbf{G}^{\top} + \mathbf{F}\mathbf{F}^{\dagger}\right)_{ij}dt. \label{eq:dYdYT}
			\end{align}
			By Lemma \ref{lem:ItoTableResult}, we have $\mathbf{O}_1,\mathbf{O}_2$ and $\mathbf{O}_3$ are diagonal matrices. Since every diagonal element of $\mathbf{O}_i , i=1,2,3$ has different creation, annihilation and conservation processes, by \Tref{tab:ItoTableResultSimplified}, requiring $d\mathbf{Y}d\mathbf{Y}^{\top}$ to be symmetric is equivalent to the symmetry of $\mathbf{G}\mathbf{O}_1\mathbf{G}^{\top}, \mathbf{G}\mathbf{O}_2\mathbf{F}^{\dagger},\mathbf{F}\mathbf{O}_3\mathbf{G}^{\top}$ and $\mathbf{F}\mathbf{F}^{\dagger}$.
			For the first term, we have $\mathbf{G}\mathbf{O}_1\mathbf{G}^{\top} = (\mathbf{G}\mathbf{O}_1\mathbf{G}^{\top})^{\top}$, which is satisfied for all $\mathbf{G}$. 
			Furthermore, for $\mathbf{F}\mathbf{F}^{\dagger}$, we have the symmetry condition
			\begin{align}
			\mathbf{F}\mathbf{F}^{\dagger} = & \mathbf{F}^{\ast}\mathbf{F}^{\top} ,
			\end{align}
			which is equivalent to $\Re(\mathbf{F})\Im(\mathbf{F})^{\top} - \Re(\mathbf{F})^{\top}\Im(\mathbf{F})=0$, and in turn equivalent to condition Eq. \eqref{eq:ConditionOnF}. For $\mathbf{G}\mathbf{O}_2\mathbf{F}^{\dagger}$, we have the symmetry condition
			\begin{align}
			\left(\mathbf{G}\mathbf{O}_2\mathbf{F}^{\dagger}\right)_{ij} = & \left(\mathbf{G}\mathbf{O}_2\mathbf{F}^{\dagger}\right)_{ji} = \sum_{k=1}^{n}\mathbf{G}_{jk}\mathbf{O}_{2,kk}\mathbf{F}^{\dagger}_{ki}.
			\end{align}
			Since every diagonal element of $\mathbf{O}_2$ will have a different creation process at each $k$,this condition is equivalent to the equality being satisfied for every $k$, which is equivalent to the condition $\mathbf{G}\mathbf{F}^{\dagger} = \mathbf{F}^{\ast}\mathbf{G}^{\top}$. This equality is equivalent to the symmetry of $\Re(\mathbf{G})\Re(\mathbf{F})^{\top}$ and $\Im(\mathbf{G})\Im(\mathbf{F})^{\top}$.
			Using a similar argument, the third line of Eq. \eref{eq:dYdYT} is also equivalent to the condition $\mathbf{G}\mathbf{F}^{\top} = \mathbf{F}\mathbf{G}^{\top}$, but this equality is equivalent to the symmetry of $\Re(\mathbf{G})\Im(\mathbf{F})^{\top}$ and $\Im(\mathbf{G})\Re(\mathbf{F})^{\top}$, which completes the proof.
						
			\begin{table}[!h]
				\centering
				\begin{tabular}{|c|c|c|c|}
					\hline  $\times$ &$d\mathbf{b}_1^{\top}$&$d\mathbf{b}_2^{\top}$&$d\mathbf{a}_1^{\top}$\\ 
					\hline  $d\mathbf{b}_1$& $\left(\mathbf{S}d\mathbf{\Lambda}^{\top}\mathbf{S}^{\dagger}\right)_{ii}$ & $\left(\mathbf{S}^*d\mathbf{A}^*\mathbf{L}^{\top}\right)_{ii}$  & $\left(\mathbf{S}^*d\mathbf{A}^*\right)_i $   \\ 
					\hline  $d\mathbf{b}_3$& $\left(\mathbf{S}d\mathbf{A}\mathbf{L}^{\dagger}\right)_{ii}$& $ \left(\mathbf{L}^{*}\mathbf{L}^{\top} dt\right)_{ii}$  & $\left(\mathbf{L}dt\right)_i$ \\
					\hline  $d\mathbf{a}_2$& $\left(\mathbf{S}d\mathbf{A}\right)_i $    & $\left(\mathbf{L}^*dt\right)_i$  & $dt$  \\
					\hline 
				\end{tabular}
				\caption{It\^{o} multiplication table for $d\mathbf{Y}$ components. } \label{tab:ItoTableResultSimplified} 
			\end{table}
			\end{proof}

			To clarify this result, we provide a few examples. In the case of a quantum system with two output channels, both subject to homodyne detection, $\mathbf{F} = \mathbf{I}$ and $\mathbf{G} = \mathbf{0}$. The case of photon counting measurement at both output channels is given by $\mathbf{F} = 0$ and $\mathbf{G} = \mathbf{I}$. A combination of homodyne detection  and photon counting measurement is given by 			
			\begin{align*}
			\mathbf{F} = \begin{pmatrix}
			1 & 0\\0 & 0
			\end{pmatrix} ,\; \mathbf{G} = \begin{pmatrix}
			0 & 0\\0 & 1
			\end{pmatrix}. 
			\end{align*} In these cases the self-commutativity condition of Theorem \ref{thm:CommutativityOfQuantumNetwork} can be easily verified. However, taking 
			\begin{align*}
			\mathbf{F} = \begin{pmatrix}
			0 & 0\\1 & 0
			\end{pmatrix} ,\; \mathbf{G} = \begin{pmatrix}
			1 & 0\\0 & 0
			\end{pmatrix}, 
			\end{align*}
			means that the first measurement is homodyne detection on the first output channel, and the second is the photon counting measurement on the same channel. Now, $\mathbf{F}\mathbf{G}^{\top}$ is not symmetric, thus by Theorem \ref{thm:CommutativityOfQuantumNetwork} the measurement vector is not self-commutative. 
			In the next subsection, we will present our second result, which gives a general derivation of a quantum filter for a set of commutative measurements.
	\subsection{General Quantum Filter For Multiple Compatible Measurements}
		To derive a quantum filter for multiple measurements, we follow the characteristic function method described in Refs. \cite{van2005feedback,gough2011quantum}. We will use the following notation to denote the conditional expectation,
			\begin{align}
			\pi(X) _t & = \hat{X}_t =   \mathbb{E}_{\rho_0 \otimes \Phi} \left[X_t | \mathcal{Y}_t\right],
			\end{align}
		where $\rho_0$ is the initial system density matrix, $\Phi$ is the vacuum state of the field, and $\mathcal{Y}_t$ is a commutative von Neumann algebra generated by measurements $\mathbf{Y}_t$.
		\begin{theorem}
			Let $\lbrace Y_{i,t}, i=1,\cdots N \rbrace $ be a set of $N$ compatible measurement outputs for a quantum system $\mathcal{G}$. With vacuum initial state, the corresponding joint measurement quantum filter is given by
			
			\begin{align}
				d\hat{X} = & \pi_t \left[ -i \left[X_t,H_t\right]  + \mathcal{L}_L(X_t)\right] dt+ \sum\limits_{i=1}^{N} \beta_{i,t} dW_{i,t},
			\end{align}	
			where $dW_{i,t} = dY_{i,t} - \pi_t\left(dY_{i,t}\right)$ is a martingale process for each measurement output and $\beta_{i,t}$ is the corresponding gain given by
			\begin{subequations}
					\begin{align}
					\beta =& \Sigma^{-1}\zeta , \label{eq:beta}\\
					\zeta^{\top}  = & \pi_t\left(X_t d\mathbf{Y}_t^{\top}\right) - \pi_t\left(X\right) \pi_t \left(d\mathbf{Y}_t^{\top}\right) + \pi_t \left(\left[\mathbf{L}^{\dagger}_t,X_t\right]\mathbf{S}_td\mathbf{A}d\mathbf{Y}_t^{\top}\right),\label{eq:zeta}\\
					\Sigma =& \pi_t \left(d\mathbf{Y}_td\mathbf{Y}_t^{\top}\right),\label{eq:Sigma} 
					\end{align}
			\end{subequations}\label{thm:JointMeasurementQuantumFiltration}
			where $\Sigma$ is assumed to be non-singular.	
		\end{theorem}
		\begin{proof}
			First, define a $\mathcal{Y}_t$-measurable It\^{o} exponential $c_{f_t}$ with respect to arbitrary functions $\lbrace f_{i,t} \rbrace$ whose derivative is given by,
			\begin{align}
			dc_{f_t} = & c_{f_t} \left[\sum\limits_{i=1}^{N}  f_{i,t}dY_{i,t} \right] = c_{f_t} d\mathbf{Y}_t^{\top}\mathbf{f}_t. \label{eq:dc_f}
			\end{align}
			Now the dynamics of the conditional expectation are assumed to be in the form of the following equation,
			\begin{align}
			d\hat{X} = & \alpha_t dt + \beta_t^{\top}d\mathbf{Y}_t ,\label{eq:dhatX}
			\end{align}
			where $\alpha_t$ and $\beta_{i,t}$ are to be determined from the conditional expectation relation
			\begin{align}
			\mathbb{E}_{\rho_0 \otimes \Phi} \left[X_t c_{f_t} | \mathcal{Y}_t \right] & = \mathbb{E}_{\rho_0 \otimes \Phi} \left[\mathbb{E}_{\rho_0 \otimes \Phi}\left[X_t | \mathcal{Y}_t\right]c_{f_t}\right], \nonumber\\
			\pi_t \left(X_t c_{f_t} \right) & = \pi_t \left(\hat{X}_t c_{f_t}\right). \label{eq:ConditionalExpectation}
			\end{align}
			\\
			From the QSDE of a system observable in Eq. \eref{eq:QSDE_X}, Eq. \eref{eq:dc_f}, and the definition of conditional expectation \eref{eq:ConditionalExpectation}, we have 
			\begin{subequations}
				\begin{align}
				d\pi_t  \left[X_t c_{f_t}\right] = &  \pi_t \left(dX_t c_{f_t} + X_t d c_{f_t} + dX_t  d c_{f_t}\right) ,\label{eq:dE_Xt_c}\\
				\pi_t  \left[dX_t c_{f_t}\right] = &  \pi_t \left[ -i \left[X_t,H_t\right] + \mathcal{L}_L(X_t)\right] c_{f_t} dt ,\label{eq:E_dXt_c}\\
				\pi_t  \left[X_t d c_{f_t}\right] = &  c_{f_t} \pi_t\left(X_t d\mathbf{Y}_t^{\top}\right)\mathbf{f}_t  ,\label{eq:E_Xt_dc}\\
				\pi_t  \left[d X_t d c_{f_t}\right] = &  c_{f_t} \left(\left[\mathbf{L}^{\dagger}_t,X_t\right]\mathbf{S}_td\mathbf{A}d\mathbf{Y}_t^{\top}\right)\mathbf{f}_t ,\label{eq:E_dXt_dc}
				\end{align}\label{eq:dE_Xt_c_All}
			\end{subequations}
			while $d\pi_t \left[\hat{X}_t c_{f_t}\right]$ given by,
			\begin{subequations}
				\begin{align}
				d\pi_t \left[\hat{X}_t c_{f_t}\right] = &  \pi_t \left(d\hat{X}_t c_{f_t} + \hat{X}_t d c_{f_t} + d\hat{X}_t  d c_{f_t}\right),\label{eq:dE_Xhatt_c}\\
				\pi_t  \left[d\hat{X}_t c_{f_t}\right] = &  c_{f_t}\left[\alpha_t dt + \beta_t^{\top} \pi_t \left(d\mathbf{Y}_t\right)\right],\label{eq:E_dXhatt_c}\\
				\pi_t  \left[\hat{X}_t d c_{f_t}\right] = &  \pi_t\left(X\right)c_{f_t}  \pi_t \left(d\mathbf{Y}_t^{\top}\right)\mathbf{f}_t,\label{eq:E_Xhatt_dc}\\
				\pi_t  \left[d \hat{X}_t d c_{f_t}\right] = &   c_{f_t} \beta_t^{\top}\pi_t \left(d\mathbf{Y}_td\mathbf{Y}_t^{\top}\right)\mathbf{f}_t.\label{eq:E_dXhatt_dc}
				\end{align}\label{eq:dE_Xhattt_c_All}
			\end{subequations}
			Equating \eref{eq:E_dXt_c} and \eref{eq:E_dXhatt_c}, solving for $\alpha_t$ and then substituting the result into Eq. \eref{eq:dhatX}, we obtain,
			\begin{align}
			d\hat{X} = & \pi_t \left[ -i \left[X_t,H_t\right] + \mathcal{L}_L(X_t)\right] dt+ \beta_t^{\top} \left[d\mathbf{Y}_t - \pi_t \left(d\mathbf{Y}_t\right)\right],\nonumber\\
			 = & \pi_t \left[ -i \left[X_t,H_t\right]  + \mathcal{L}_L(X_t)\right] dt+ \beta_t^{\top} d\mathbf{W}_t.
			\end{align}
			Furthermore, using the fact that the function $\mathbf{f}_t$ is arbitrary, we can equate the right-hand-side of Eqs. \eref{eq:dE_Xt_c_All} and \eref{eq:dE_Xhattt_c_All}, which recovers $\beta_{i,t}$,
			\begin{align}
			\beta^{\top}_t =& \zeta^{\top} \mathbf{\Sigma} ^{-1},
			\end{align}
			where $\zeta$ and $\Sigma$ are a real valued row vector and a real valued matrix, respectively given by,
			\begin{subequations}
				\begin{align}				
				\zeta^{\top}  = & \pi_t\left(X_t d\mathbf{Y}_t^{\top}\right) - \pi_t\left(X\right) \pi_t \left(d\mathbf{Y}_t^{\top}\right) + \pi_t \left(\left[\mathbf{L}^{\dagger}_t,X_t\right]\mathbf{S}_td\mathbf{A}d\mathbf{Y}_t^{\top}\right),\\
				\Sigma =& \pi_t \left(d\mathbf{Y}_td\mathbf{Y}_t^{\top}\right).
				\end{align}
			\end{subequations}
			Proof of the martingale properties of $W_{i,t}$	is given in Ref. \cite[proof of Theorem 7.1]{bouten2007introduction}.
		\end{proof}
		A restricted form of Theorem \ref{thm:JointMeasurementQuantumFiltration} has been independently proven in Ref. \cite[Theorem 9]{nurdin2014quantum} using the "reference probability" approach. This result applied to a class of generalized homodyne detection measurements, i.e. $\mathbf{G} = 0$.
		\\
		The result of Theorem \ref{thm:JointMeasurementQuantumFiltration} can be generalized further to include coherent initial states $\psi(u)$ by introducing a Weyl operator in Eq. \eref{eq:WeylOperator}. To do this, we select $U= I$ and $u\neq0$, in the Weyl operator parameters, and transform all of the fundamental quantum processes $M_t$ in \Eref{eq:QuantumFundamentalProcesses} into $M_t(u) = W(u,I)^{\dagger}M_t W(u,I)$. \cite{bouten2007introduction,parthasarathy2012}. 
		\\
		The dynamics of the quantum filter can also be expressed using the following equation,
		\begin{align}
		d\hat{X} = & \pi_t \left[ -i \left[X_t,H_t\right] + \mathcal{L}_L(X_t)\right] dt + \zeta^{\top} \mathbf{\Sigma} ^{-1}d\mathbf{W}. \label{eq:dHatXGeneral}
		\end{align} 
		From a classical filtering point of view, Eq. \eref{eq:dHatXGeneral} possesses some similarities to the Kalman filter, where $\pi_t \left[ -i \left[X_t,H_t\right] + \mathcal{L}_L(X_t)\right] dt$ is the a-priori estimate and $\zeta^{\top} \mathbf{\Sigma} ^{-1}$ is analogous to the Kalman gain which multiplies the innovation process $d\mathbf{W}$. 
		\begin{remark}
		Theorem \ref{thm:JointMeasurementQuantumFiltration} requires the existence of an invertible  differential measurement correlation matrix $\Sigma$, which is a sufficient condition for the joint measurements to be obtainable from a single quantum filter equation. This condition, however, is not a necessary condition, as we will encounter in Section \ref{sec:Limit}, where in a case of zero reflectivity, the quantum filter equation exists although $\Sigma$  is not invertible. 
		\end{remark}
		 
		In most cases of nonlinear estimation, Eq. \eref{eq:dHatXGeneral} is merely a representation for the estimator and cannot be interpreted as an explicit solution to the filtering problem \cite{segall1975nonlinear}. As in the classical filtering problem, explicit solutions to the general nonlinear filtering problem can be obtained using a variety of approximation methods \cite{lototsky2006wiener,crisan2011oxford}. However, in the quantum filter, rather than approximating the explicit solution of Eq. \eref{eq:dHatXGeneral}, one can convert the filtering problem in Eq. \eref{eq:dHatXGeneral}, which is given in the Heisenberg picture, into the Schr\"{o}dinger picture. Then one deals with the evolution of the system's conditional density operator at time $t$, $\rho_t$. As described in Ref. \cite{bouten2004stochastic}, by means of the relation $\pi_t(X) \equiv \text{tr}(\rho_t X)$, one can construct from Eq. \eref{eq:dHatXGeneral},
		\begin{align}
		d\rho_t = & \underbrace{\left[-i\left[H_t,\rho_t\right] + \mathbf{L}^{\top}\rho_t\mathbf{L}^{\ast} - \frac{1}{2}\mathbf{L}^{\dagger}\mathbf{L}\rho_t - \frac{1}{2}\rho_t\mathbf{L}^{\dagger}\mathbf{L}\right]dt }_{\text{a-priori}}+ \underbrace{\zeta_{\rho}^{\top}\Sigma^{-\top}d\mathbf{W}}_{\text{innovation term}},\label{eq:GeneralFilteringEquationSchrodinger}
		\end{align}
		where $\zeta_{\rho}$ in the above equation is now only a function of the conditional density operator, $\mathbf{L}$ and $H$, but not of the particular system observable $X$. 
		\\
		Finally, for numerical simulation efficiency, after truncating the Hilbert space dimension to a finite number $n$, instead of solving for the $ n \times n$ conditional density operator  in Eq. \eref{eq:GeneralFilteringEquationSchrodinger}, one can "\emph{unravel}" this equation, and solve instead for the state vector $|\psi \rangle$, which is an $n \times 1$ vector. This unravelled equation is of the form
		\begin{align}
		d |\psi\rangle = & -i \left(H_t - i \dfrac{1}{2}\mathbf{L}^{\dagger}\mathbf{L} \right)|\psi\rangle dt + \sigma(|\psi\rangle) dt + \delta(|\psi\rangle) d\mathbf{W}, \label{eq:GeneralFilteringEquationSchrodingerUnravel}
		\end{align}
		where $\sigma(|\psi\rangle)$ and $\delta(|\psi\rangle)$ are operators which are linear with respect to the coupling operator $\mathbf{L}$.
	\section{Application of The Quantum Filter to Joint Homodyne Detection and Photon Counting}
	\subsection{Quantum Filter for Joint Homodyne Detection and Photon Counting}\label{sec:Limit}
		\begin{figure}
			\centering
			\includegraphics[width=0.5\textwidth]{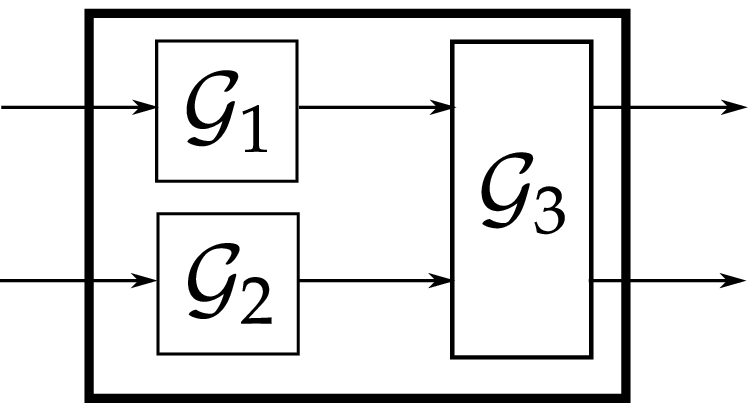}
			\caption{Quantum network depiction of the quantum optics setup of \Fref{fig:DoubleMeasurement}}
			\label{fig:QuantumNetwork}
		\end{figure}
		
	In this section, we derive the quantum filter for the setup shown in \Fref{fig:DoubleMeasurement}. We define two Fock spaces for the two corresponding input fields, the first Fock space for the system input field is given by $\Gamma_1(\mathsf{h})$, while the vacuum field input at the upper-port of the beam splitter is denoted by $\Gamma_2(\mathsf{h})$. Notice that if we restrict ourself to work in the linear span of coherent states, the Fock spaces $\Gamma_{i}$ $i=1,2$ possesses a continuous tensor product. For any time interval $0\leq s < t $, the Fock space $\Gamma_i$ can be decomposed into 
	\begin{equation}
	\Gamma_i = \Gamma_{i,s]}\otimes \Gamma_{i,[s,t]}\otimes \Gamma_{i,[t}. \label{eq:GammaDecompose}
	\end{equation}
	\\
	The overall quantum system with the measurement setup in \Fref{fig:DoubleMeasurement} can be depicted as shown in \Fref{fig:QuantumNetwork}. $\mathcal{G}_1$ is our system of interest, with parameters $(I,L,H)$. We concatenate the vacuum noise into our system by introducing a second open quantum system, $\mathcal{G}_2$ , whose parameters are $(1,0,0)$. The last open quantum system $\mathcal{G}_3$ is the beam splitter, with parameters $(\mathbf{S},0,0)$. The parameters of the composite quantum system are obtained by taking the series product and the concatenation product \cite{gough2009series}, giving $\mathcal{G} = \left(\mathcal{G}_1 \boxplus \mathcal{G}_2\right)\rhd\mathcal{G}_3$ with parameters $\left(\mathbf{S},\mathbf{S}\begin{pmatrix}
	L\\0
	\end{pmatrix},H\right)$.
	\\
	The output field of the system $\mathcal{G}_1$, $A_{s,t} = U^{\dagger}_t\left(I \otimes A_{i,t} \right)U_t$, is an operator on $\mathsf{h}_s\otimes\Gamma_{1,t]}(\mathsf{h})$, while the vacuum field $A_{v,t}$ is an operator on $\Gamma_{2,t]}(\mathsf{h})$. We denote the total Hilbert Space as $\mathsf{H} = \mathsf{h}_s\otimes\Gamma_1(\mathsf{h})\otimes\Gamma_2(\mathsf{h})$. The beam splitter equation is given by,
	\begin{align}
	\mathbf{S} & = \begin{pmatrix}
	\sqrt{1 - r^2} e^{i\theta}& re^{i\left(\theta + \frac{\pi}{2}\right)}\\
	re^{i\left(\theta + \frac{\pi}{2}\right)} & \sqrt{1 - r^2} e^{i\theta}
	\end{pmatrix},& r \geq& 0. \label{eq:BeamSplitterGeneral}
	\end{align}
	For homodyne measurement in the first output channel and photoncounting measurement in the second output channel, we have
		\begin{align*}
		d\mathbf{Y}_t =& \mathbf{F}^{\ast}d\tilde{\mathbf{A}}_t^{\ast} +\mathbf{F}d\tilde{\mathbf{A}}_t+ \mathbf{G}d\tilde{\mathbf{\lambda}}_t, \\
		\mathbf{F} =& \begin{pmatrix}
		1 & 0\\ 0 & 0
		\end{pmatrix} ,\; 
		\mathbf{G} = \begin{pmatrix}
		0 & 0\\ 0 & 1
		\end{pmatrix}.
		\end{align*}
		\\
	By Theorem \ref{thm:CommutativityOfQuantumNetwork}, the measurement set $d\mathbf{Y}$ is self-commutative. Substituting the general beam splitter \eref{eq:BeamSplitterGeneral} and output field evolution \eref{eq:OutputFieldEvolution}, the measurements quantum stochastic differential equations are given by,
	\begin{subequations}
	\begin{align}
	dY_{1,t} = &  \sqrt{1 - r^2} \left( \left(e^{i\theta}L_t + e^{-i\theta}L_t^{\dagger}\right) dt + e^{i\theta}dA_{i,t} + e^{-i\theta}dA^{\dagger}_{i,t}\right) \nonumber \\
	& + ir\left(e^{i\theta}dA_{v,t} - e^{-i\theta}dA^{\dagger}_{v,t}\right),\label{eq:dY_1General}\\
	dY_{2,t} = & r^2\left[ d\Lambda_{i,t} + L_t dA^{\dagger}_{i,t} + L^{\dagger}_t dA_{i,t} +  L^{\dagger}_t L_t dt\right]+ \left(1-r^2\right)d\Lambda_{v,t} \nonumber\\
	& + i \left(r \sqrt{1-r^2}\right)\left[d\Lambda_{vi,t} - d\Lambda_{iv,t} + L_t dA^{\dagger}_{v,t} - L_t^{\dagger}dA_{v,t}\right].\label{eq:dY_2General}
	\end{align}
	\end{subequations}
	These measurements can be proven to satisfy the non-demolition and self-non-demolition properties, see Ref. \cite{van2005feedback}. Next, we can compute the expectation and the correlation of the measurement time derivative as
	\begin{subequations}
	\begin{align}
	\pi_t \left(dY_{1,t}\right) = & \sqrt{1 - r^2} \pi_t\left(e^{i\theta}L_t + e^{-i\theta}L_t^{\dagger}\right) dt, \\
	\pi_t \left(dY_{2,t}dY_{2,t}\right) = & \pi_t \left(dY_{2,t}\right) =  r^2\pi_t\left(L^{\dagger}_t L_t\right) dt, \\
	\pi_t \left(dY_{1,t}dY_{1,t}\right) = & dt, \\
	\pi_t \left(dY_{2,t}dY_{1,t}\right) = & \pi_t \left(dY_{1,t}dY_{2,t}\right) =  0.
	\end{align}\label{eq:dY1dY2ExpectationGeneral}
	\end{subequations}
	Using these values, $\beta$ is then given by
	\begin{align}
	\beta_1 = & \sqrt{1 - r^2} \left(\pi_t\left(X_t e^{i\theta} L_t + e^{-i\theta} L^{\dagger}_t X_t\right) - \pi_t\left(X\right)\pi_t\left(e^{i\theta}L_t + e^{-i\theta} L^{\dagger}_t\right)\right),\\
	\beta_2 = & \dfrac{\pi_t\left( L^{\dagger}_t X_t L_t \right)}{\pi_t\left(L^{\dagger}_t L_t\right)}- \pi_t\left(X\right).
	\end{align}
	In the case that $r \rightarrow 1$, the estimation problem reduces to an estimation problem with a single photon counting process. The opposite case is more interesting. When  $r \rightarrow 0$, the gain $\beta_2$ has a non zero value, while the Poisson process has zero arrival rate, and hence the estimation problem reduces to an estimation problem with a single homodyne detection. This is unsurprising since zero reflection ensures all photons pass through to the homodyne detector. 
	\\
	We can unravel the stochastic master equation into the form Eq. \eref{eq:GeneralFilteringEquationSchrodingerUnravel}. By using the It\^{o} equivalence,
	\begin{align}
	d\rho &= d|\psi\rangle \langle \psi| + |\psi\rangle d\langle \psi| + d|\psi\rangle d\langle \psi|,
	\end{align}
	one recovers the unravelled stochastic Schr\"{o}dinger equation for the quantum filter,
	\begin{subequations}
	\begin{align}
	d |\psi\rangle = & -i \left(H - i \dfrac{1}{2}L^{\dagger}L \right)|\psi\rangle dt + \sigma(|\psi\rangle) dt + \delta_1(|\psi\rangle) dW + \delta_2(|\psi\rangle) dN, \label{eq:QFUnravelGeneral}\\
	\sigma(|\psi\rangle) = & \left(\left[\dfrac{1 - r^2}{2}\pi_t\left(e^{-i\theta}L^{\dagger}_t+e^{i\theta}L_t\right)\right]\right.L \nonumber\\
	& \left. +\left[\dfrac{r^2\pi_t\left(L^{\dagger}_t L_t\right)}{2}-\dfrac{1-r^2}{8}\pi_t\left(e^{-i\theta}L^{\dagger}_t+e^{i\theta}L_t\right)^2\right]\right)|\psi\rangle,\\
	\delta_1(|\psi\rangle) = & \sqrt{1-r^2}\left(L - \dfrac{1}{2}\pi_t\left(e^{-i\theta}L^{\dagger}_t+e^{i\theta}L_t\right)\right)|\psi\rangle,\\
	\delta_2(|\psi\rangle) = & \left(\dfrac{L}{\sqrt{\pi_t\left(L^{\dagger}_t L_t\right)}} - 1\right)|\psi\rangle.
	\end{align}\label{eq:QFUnravel}
	\end{subequations}
	Here, $dW$ is equal to $dW_1$, and $dN$ is equal to the Poisson process of the second measurement. The unravelled version of quantum filter given in Eq. \eref{eq:QFUnravel} is normalized. For the case $r=0$ and $r=1$, Eq. \eref{eq:QFUnravel} is equivalent to stochastic Schr\"{o}dinger equation (SSE) for homodyne detection and photon counting respectively, given in Refs. \cite[\textsection 6.1 \textsection 6.4]{breuer2007theory},\cite[\textsection 11.3 \textsection 11.4]{gardiner2004quantum}. 
	\\
	\subsection{Comparison with results in Ref. \cite{Kuramochi2013}}
	In this subsection, we give a comparison of our quantum filter with the results of Ref. \cite{Kuramochi2013}. Here, the unnormalized SSE for photon counting and homodyne detection was formulated heuristically by the addition of two measurement operations, where every operation determined the infinitesimal evolution of the unnormalized state. The SSE for photon counting and homodyne was given in Ref. \cite{Kuramochi2013} as, 
	\begin{subequations}
	\begin{align}
	|\tilde{\psi}_{t+dt}\rangle = & \left[1 + A dt + \left(B-1\right)dN + C dW\right]|\tilde{\psi}_t\rangle, \\
	A &=  -i H - \dfrac{L^{\dagger}L}{2} + L \langle L + L^{\dagger}\rangle ,\\
	C=B &=  L.
	\end{align}\label{eq:UnNormalizedKuramochi}
	\end{subequations}
	In these equations, we are slightly abusing the notation, by denoting $\pi_t\left(X\right) = \langle X\rangle$, and setting local oscillator angle to $\theta=0$.	To give a comparison of Eq. \eref{eq:UnNormalizedKuramochi} with our result in Eq. \eref{eq:QFUnravel}, one can consider the normalization of Eq. \eref{eq:UnNormalizedKuramochi} as detailed in Ref. \cite[\textsection 11.4]{gardiner2004quantum}. In general, the infinitesimal evolution given in Eq. \eref{eq:UnNormalizedKuramochi}, can be normalized to the following normalized SSE,
	\begin{subequations}
		\begin{align}
		d|\psi_t\rangle = & \left[\left(A+\hat{A}+C\hat{C}\right) dt + \left(\hat{B}B-1\right)dN + \left(C+\hat{C}\right) dW\right]|\psi_t\rangle, \\
		\hat{A} &=  \frac{3}{8}\langle C + C^{\dagger}\rangle^2 - \frac{1}{2} \langle A + A^{\dagger}\rangle - \frac{1}{2} \langle C^{\dagger}C\rangle,\\
		\hat{B} &=  \dfrac{1}{\sqrt{ \langle B^{\dagger} B \rangle} },\\
		\hat{C} &=  -\frac{1}{2}\langle C + C^{\dagger}\rangle.
		\end{align}\label{eq:NormalizedKuramochi}
	\end{subequations}
	Substituting these values into Eq. \eref{eq:UnNormalizedKuramochi}, one can get an SSE in the form of Eq. \eref{eq:QFUnravelGeneral}, with 
	\begin{subequations}
		\begin{align}
		\sigma(|\psi\rangle) = & \left[ \frac{1}{2} \langle L^{\dagger} + L\rangle L - \frac{1}{8} \langle L^{\dagger}+ L\rangle^2 \right] |\psi\rangle, \\
		\delta_1(|\psi\rangle) = & \left[ L - \frac{1}{2} \langle L^{\dagger}+ L\rangle \right] |\psi\rangle, \\
		\delta_2(|\psi\rangle) = &  \left[\dfrac{L}{\sqrt{\langle L^{\dagger} L\rangle}} - 1\right]|\psi\rangle.
		\end{align}	\label{eq:NormalizedKuramochiDetail}
	\end{subequations}
	However Eq. \eqref{eq:NormalizedKuramochiDetail} is \emph{not} consistent with our result in Eq. \eref{eq:QFUnravelGeneral}. In fact, Eq. \eref{eq:NormalizedKuramochiDetail} is consistent with Eq. \eref{eq:QFUnravelGeneral} in the limiting case $r=0$, which would correspond to a jump process with zero arrival rate. The paper \cite{Kuramochi2013} claims to consider the case of simultaneous jump and diffusion measurement processes, but the above comparison shows that it does not account for the required beam splitter. The equivalence of the unnormalized SSE in Ref. \cite{Kuramochi2013}  with the quantum filter \eref{eq:QFUnravel} is obtained when we take the beam splitter into consideration. In this case, instead of $A,B$ and $C$ given in Eq. \eref{eq:UnNormalizedKuramochi}, we will have an equivalent quantum filter as an unnormalized SSE given by
	
	\begin{subequations}
		\begin{align}
		A &=  -i H - \dfrac{L^{\dagger}L}{2} + \left(1 - r^2\right)L \langle L + L^{\dagger}\rangle,\\
		B &=  r L,\\
		C &=  \sqrt{1 - r^2}L.
		\end{align}\label{eq:NonNormalizedQF}
	\end{subequations}
	
	\Eref{eq:NonNormalizedQF} gives an intuitive interpretation of simultaneous filtering, where as in Ref. \cite{Kuramochi2013}, the unnormalized evolution of photon counting and homodyne detections requires the addition of the two measurement operations, but with the appropriate beam splitter gain.

	\section{Simulation Results}
	This section will show a simulation of the proposed quantum filter for an empty cavity with a number state as the initial condition. In this condition, the analytical probability distribution of the number state is given by \cite{breuer2007theory},
	\begin{subequations}
	\begin{align}
		\mathbf{P}_N(t) =& {{n}\choose{N}} \mu(t)^N - \left(1-\mu(t)\right)^{n-N},\\
		\mu(t) =& 1- e^{-\gamma t}. 
		\end{align}	\label{eq:Analytic}
	\end{subequations}
	Simulation results for different reflectivity factors are shown in \Fref{fig:EmptyCavity}. The cavity dissipative ratio $\gamma$ is set to one. 
	Figures \ref{fig:r2_1} and \ref{fig:r2_0} show single trajectory simulations of the expected number operator for the case of pure photon counting measurement and homodyne detection. \Fref{fig:r2_0.5_100} shows the non-trivial case of a half-reflective beam splitter $r^2 = 0.5$. In this case, the quantum filter average of the number operator converges to the analytical  prediction \eref{eq:Analytic} when the trial number is increased. \Fref{fig:r2_0.5_100} also shows that the SSE formulated in \cite{Kuramochi2013} gives a biased average compared with the analytical result.
	\begin{figure}[!h]
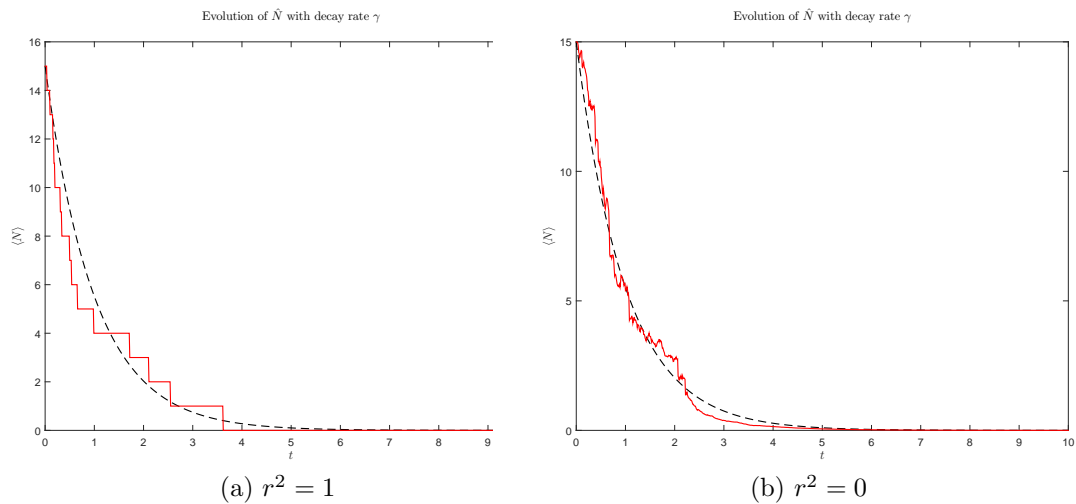

	\centering
		\subfloat[$r^2=1$]
		{\label{fig:r2_1}\includegraphics[width=0.45\textwidth]{Empty_Result1}}
		\subfloat[$r^2=0$]
		{\label{fig:r2_0}\includegraphics[width=0.45\textwidth]{Empty_Result2}}
	\caption{Single trajectory Monte-Carlo realizations of the quantum filter with initial number state and dissipation, with beam splitter reflectivity such that (a) $r^2 = 1$ and (b) $r^2 = 0$.} 
	\label{fig:EmptyCavity}
	\end{figure}

		\begin{figure}[!h]
		\centering
		\input{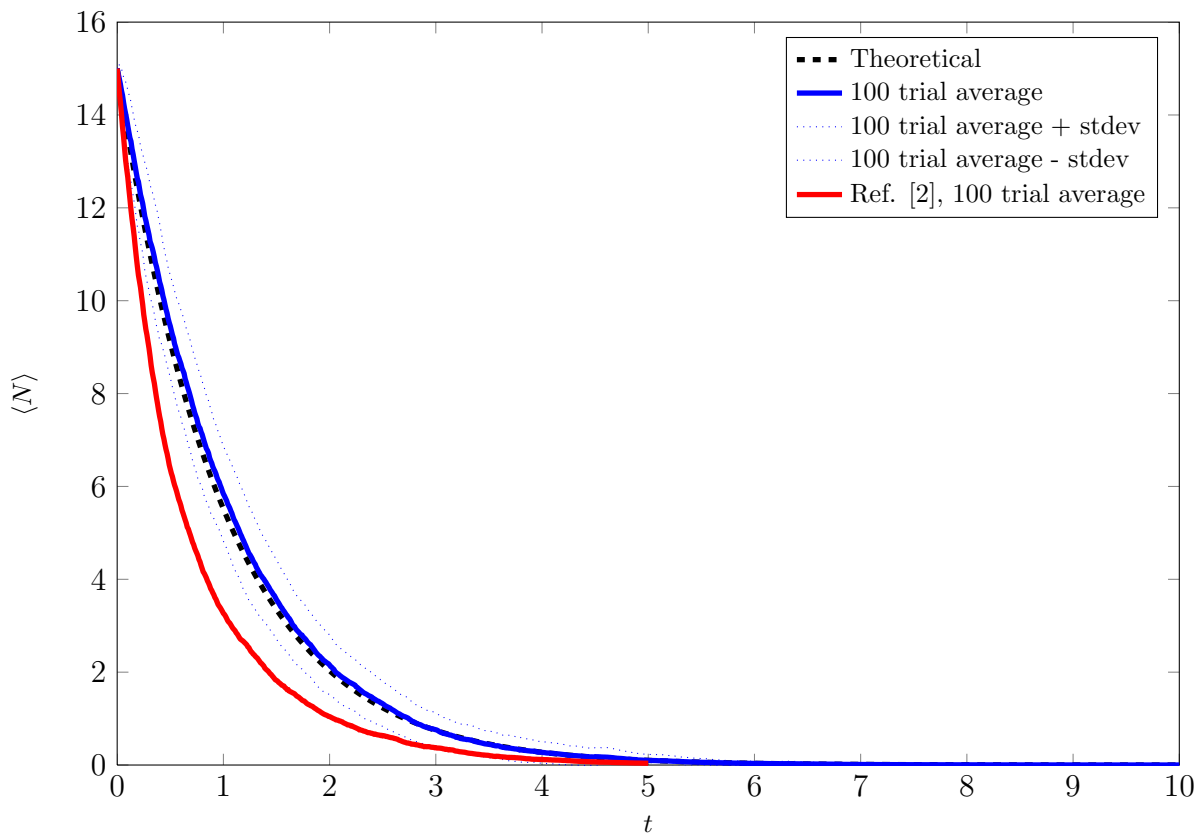}
		\caption{Expected of number operator as a function of time with number state initial condition and dissipation. The average of 100 Monte-Carlo trials along with a comparison to the analytical results and result of Ref. \cite{Kuramochi2013}. This figure shows the case of half reflective beam splitter $r^2 = 0.5$. The quantum filter expected number operator converges to the analytical prediction of Eq. \eref{eq:Analytic}. \Fref{fig:r2_0.5_100} also shows that the SSE formulated in Ref. \cite{Kuramochi2013} gives a biased expectation compared with the analytical result.}
		\label{fig:r2_0.5_100}
		\end{figure}

	\section{Conclusions}
	We have derived a sufficient and necessary condition for a class of quantum measurement output channels to satisfy a commutativity relation. The measurement class considered is quite general compared to Ref. \cite{nurdin2014quantum}, since it covers not only homodyne type measurements, but also photon counting type measurements. Furthermore, this commutativity condition enables us to derive a quantum filter corresponding to multiple measurement outputs. We also provide examples of the quantum filter for homodyne and photon counting detection. The quantum filter results were shown to be consistent with the homodyne and photon counting quantum filters for both extreme cases, where the reflectivity of the beam splitter is zero and one. In addition, the quantum filter also corrected the result of Ref. \cite{Kuramochi2013}, which ignored the effect of the beam splitter.
	\section{Acknowledgements}
	We acknowledge discussions with Dr. Katanya Kuntz of UNSW Canberra.

	\section{References}
	
	\bibliographystyle{iopart-num}
	\bibliography{Reference}
	
\end{document}